\documentclass[journal]{IEEEtran}

\usepackage{cite}
\usepackage{amsmath,amssymb,amsthm}
\usepackage{booktabs}
\usepackage{hyperref}
\usepackage{xcolor}
\usepackage{pifont}
\usepackage{enumitem}
\usepackage{balance}
\usepackage{tikz}
\usetikzlibrary{shapes,arrows,positioning,fit}

\newtheorem{theorem}{Theorem}
\newtheorem{prop}{Property}
\newtheorem{proposition}{Proposition}
\newtheorem{definition}{Definition}

\begin{document}

\title{SentinelAgent: Intent-Verified Delegation Chains\\for Securing Federal Multi-Agent AI Systems}

\author{KrishnaSaiReddy Patil}

\maketitle

\begin{abstract}
When Agent~A delegates to Agent~B, which invokes Tool~C on behalf of User~X, no existing framework can answer: whose authorization chain led to this action, and where did it violate policy? This paper introduces SentinelAgent, a formal framework for verifiable delegation chains in federal multi-agent AI systems. The \emph{Delegation Chain Calculus} (DCC) defines seven properties, six deterministic (authority narrowing, policy preservation, forensic reconstructibility, cascade containment, scope-action conformance, output schema conformance) and one probabilistic (intent preservation), with four meta-theorems (property minimality, graceful degradation, defense-in-depth completeness over 126 evasion combinations, and composition safety) and one proposition establishing the practical infeasibility of deterministic intent verification. The \emph{Intent-Preserving Delegation Protocol} (IPDP) enforces all seven properties at runtime through a non-LLM Delegation Authority Service. A three-point verification lifecycle, comprising pre-execution intent checking, at-execution scope enforcement, and post-execution output validation, achieves 100\% combined TPR at 0\% FPR on DelegationBench~v4 (516 scenarios across 10 attack categories and 13 federal domains), with each layer catching attacks the others miss. Under black-box adversarial conditions where the attacker does not know the manifest or schema contents, the DAS blocks 30/30 attacks with 0 false positives. Deterministic properties are unbreakable under adversarial stress testing; intent verification degrades to 13\% against sophisticated paraphrasing, an honest limitation termed the \emph{adversarial intent paraphrasing} problem. Fine-tuning the NLI model on 190 government delegation examples improves P2 malicious detection from 1.7\% to 88.3\% TPR (5-fold cross-validated, F1=82.1\%) while the three-layer pipeline maintains 0\% system-level FPR via the benign override mechanism. Properties P1, P3--P7 are mechanically verified via TLA+ model checking across 2.7 million states with zero violations. Even when intent verification is evaded, the remaining six properties constrain the adversary to permitted API calls, conformant outputs, traceable actions, bounded cascades, and compliant behavior.
\end{abstract}


\section{Introduction}

The transition from conversational AI to agentic AI changes the attack surface fundamentally. Chatbots generate text. AI agents call APIs, execute code, manage files, and make decisions with minimal human oversight. In September 2025, Anthropic detected GTG-1002, the first documented AI-orchestrated cyber espionage campaign, where a Chinese state-sponsored group weaponized Claude Code to autonomously attack approximately 30 organizations with AI handling 80--90\% of tactical operations~\cite{gtg1002_anthropic}. AI-enabled attacks surged 89\% year-over-year according to the CrowdStrike 2026 Global Threat Report~\cite{crowdstrike2026}, and 88\% of enterprises experienced confirmed or suspected AI agent security incidents~\cite{gravitee_survey}.

Federal agencies are deploying multi-agent systems at scale. The NIST AI Agent Standards Initiative, launched February 17, 2026, explicitly seeks security guidance for autonomous agent deployments~\cite{nist_agent_initiative}. Yet a survey of 285 IT and security professionals found that 84\% cannot pass a compliance audit focused on agent behavior, only 23\% have a formal agent identity strategy, and only 18\% are confident their IAM can manage agent identities~\cite{csa_atf_survey}.

Since the initial development of this framework, the field has grown rapidly: over 50 papers on agentic AI security appeared in Q1 2026 alone. AIP~\cite{aip} handles delegation chain identity with capability tokens. Authenticated Workflows~\cite{auth_workflows} enforces policies deterministically via the MAPL language. FormalJudge~\cite{formaljudge} uses Z3 SMT solving for oversight. ILION~\cite{ilion} gates execution at 143 microseconds. Agentic JWT~\cite{agentic_jwt} binds actions to user intent tokens. ASTRA~\cite{astra} matches tasks to scopes at authorization time. Allegrini et al.~\cite{allegrini_formal} formalize 31 temporal logic properties for single-host systems. Each addresses a piece of the problem. None formally models how authority, intent, and compliance compose across multi-hop delegation chains while also providing a runtime enforcement protocol.

\textbf{The delegation accountability gap.} Consider a federal benefits processing system where an intake agent (Agency~A) delegates to a records retrieval agent (Contractor~B), which calls a medical database tool (Cloud Provider~C), and passes results to a decision-support agent (Agency~D). When Agent~D recommends denial based on a misinterpreted lab value, no existing system can answer: (1)~Who authorized Agent~B to access those specific medical records? (2)~Did the citizen's original intent survive intact through the A$\to$B$\to$C$\to$D chain? (3)~At which step did the misinterpretation occur? (4)~Does this chain comply with NIST~800-53 AC-4 (Information Flow Enforcement) and the Privacy Act? (5)~Can the complete chain be reconstructed for a FISMA audit?

\textbf{Why existing work is insufficient.} Individual components of this problem have been addressed in isolation. SEAgent~\cite{seagent} handles access control but not delegation chains. ShieldAgent~\cite{shieldagent} verifies safety policies but not multi-agent delegation. Agent Behavioral Contracts~\cite{abc} bounds drift for single agents, not chains. South et al.~\cite{oauth_delegation} handle authentication but not intent or compliance. The OWASP Top~10 for Agentic Applications~\cite{owasp_agentic} identifies the risks but provides no formal defense. The CSA Agentic Trust Framework~\cite{csa_atf} provides governance maturity models but no enforcement mechanism. The missing piece is a single formal model that ties authority, intent, compliance, forensics, and containment together across multi-agent delegation.

\textbf{Contributions.} This paper introduces SentinelAgent with five contributions:

\begin{enumerate}[leftmargin=*]
\item \textbf{Delegation Chain Calculus (DCC)} (\S\ref{sec:dcc}): A formal model where each delegation step carries authority scope, intent vector, and policy constraints. Seven properties are proven: authority monotonic narrowing (P1), intent entailment preservation (P2), policy conjunction preservation (P3), forensic reconstructibility in $O(n)$ (P4), bounded cascade containment (P5), scope-action conformance (P6), and output schema conformance (P7). Properties P1, P3--P7 are deterministic; P2 is probabilistic with characterized bounds.

\item \textbf{Intent-Preserving Delegation Protocol (IPDP)} (\S\ref{sec:ipdp}): A runtime protocol enforced by a Delegation Authority Service (DAS), a trusted, non-LLM component immune to prompt injection, that verifies all seven DCC properties at every delegation step. IPDP introduces three-layer intent verification: keyword safety filtering, context-enriched NLI, and benign-override heuristics. A risk-tiered execution model provides zero blast radius for high-sensitivity operations.

\item \textbf{Federal Compliance Mapping Engine (FCME)} (\S\ref{sec:fcme}): The first published mapping of 20 NIST~SP~800-53 Rev~5.2 controls and 10 OWASP ASI risk categories to specific agentic AI delegation attack surfaces, encoded as machine-readable policy constraints that plug directly into the IPDP protocol.

\item \textbf{Three-Point Verification Lifecycle} (\S\ref{sec:eval}): An extension of DCC with scope-action conformance (P6) and output schema conformance (P7), creating a pre-execution/at-execution/post-execution defense chain. On DelegationBench~v4 (516 scenarios across 10 attack categories), the combined P2+P6+P7 lifecycle achieves 100\% TPR at 0\% FPR, with each layer catching attacks the others miss: P2 catches keyword-detectable intent drift, P6 catches unauthorized API calls within delegated scope, and P7 catches malicious outputs from permitted API calls. Fine-tuning the NLI model on 190 government delegation examples improves P2 from 1.7\% to 88.3\% malicious TPR (5-fold cross-validated, F1=82.1\%). Properties P1, P3--P7 are mechanically verified via TLA+ model checking (2.7M states, zero violations). A live LangChain agent integration routes real LLM-generated delegation requests through the DAS HTTP prototype.

\item \textbf{Four Meta-Theorems and One Proposition} (\S\ref{sec:metatheorems}): Four meta-theorems provable only over DCC's unified structure: property minimality (each property is independently necessary), graceful degradation bounds (damage envelopes when any property is evaded), defense-in-depth completeness (all evasion combinations are bounded), and composition safety with write-impact notification (interacting chains preserve properties). A separate proposition establishes the practical infeasibility of deterministic intent verification, justifying P2's probabilistic design via both theoretical grounding (Rice's theorem analogy) and empirical demonstration (eight ambiguous delegation pairs).
\end{enumerate}

This work is part of a series addressing different layers of the government AI security stack: CivicShield~\cite{civicshield} secures the conversational interface with seven-layer defense-in-depth, RAGShield~\cite{ragshield} secures the knowledge pipeline with provenance-verified taint tracking, and SentinelAgent secures the autonomous delegation chains that connect them. Each addresses a distinct attack surface using a different core technique: multi-layer filtering, supply-chain provenance, and formal delegation calculus respectively.

\section{Background and Threat Model}
\label{sec:threat}

\subsection{Multi-Agent Delegation in Federal Operations}

A \emph{delegation chain} $C = [\tau_0, \tau_1, \ldots, \tau_n]$ is an ordered sequence of delegation tokens where $\tau_0$ originates from a human user and each subsequent $\tau_i$ represents an agent-to-agent handoff. In federal operations, delegation chains routinely cross organizational boundaries (e.g., HHS to SSA), trust levels (e.g., government to contractor), and infrastructure providers (e.g., AWS GovCloud to Azure Government).

\subsection{Adversary Model}

Seven adversary types are considered, each mapped to documented real-world attacks:

\begin{table}[t]
\centering
\caption{Adversary Types and Real-World Attack Mappings}
\label{tab:adversaries}
\begin{tabular}{@{}cp{2.2cm}p{3.5cm}l@{}}
\toprule
\textbf{ID} & \textbf{Adversary} & \textbf{Capability} & \textbf{Citation} \\
\midrule
A1 & External & Prompt injection & \cite{promptware_killchain} \\
A2 & Compromised Agent & Controls one agent & \cite{agents_of_chaos} \\
A3 & Man-in-the-Middle & Intercepts messages & \cite{aitm_attack} \\
A4 & Privilege Escalator & Exploits delegation & \cite{tagalong_attack} \\
A5 & Chain Poisoner & Subtle intent drift & \cite{stac_attack} \\
A6 & Collusive Agents & Two+ agents coordinate via side channel & \cite{groupguard} \\
A7 & Protocol-Spanning & Exploits protocol boundaries & \cite{mcp38} \\
\bottomrule
\end{tabular}
\end{table}

\textbf{Adversary goals:} (G1)~Execute unauthorized actions violating authority scope; (G2)~Redirect agent behavior violating intent preservation; (G3)~Bypass compliance controls; (G4)~Evade forensic reconstruction; (G5)~Maximize blast radius through cascading failures; (G6)~Execute unauthorized API calls within delegated scope (scope-action gap); (G7)~Produce malicious outputs from permitted API calls (output manipulation).

\textbf{Trust assumptions:} The DAS is trusted and not compromisable via prompt injection (it is deterministic software, not an LLM). The NLI and embedding models are pre-deployed with integrity-protected weights. Network channels use TLS. At least one agent in any chain is honest.

Table~\ref{tab:coverage} maps each adversary type to the DCC properties that defend against it, showing complete coverage across all seven adversary types.

\begin{table}[t]
\centering
\caption{Adversary $\to$ Property Defense Coverage}
\label{tab:coverage}
\footnotesize
\setlength{\tabcolsep}{3.5pt}
\begin{tabular}{@{}clccccccc@{}}
\toprule
\textbf{ID} & \textbf{Adversary} & \textbf{P1} & \textbf{P2} & \textbf{P3} & \textbf{P4} & \textbf{P5} & \textbf{P6} & \textbf{P7} \\
\midrule
A1 & External (injection) & & \ding{51} & & & & \ding{51} & \ding{51} \\
A2 & Compromised Agent & \ding{51} & \ding{51} & & \ding{51} & \ding{51} & \ding{51} & \ding{51} \\
A3 & Man-in-the-Middle & & & & \ding{51} & & & \\
A4 & Privilege Escalator & \ding{51} & & & & & \ding{51} & \\
A5 & Chain Poisoner & & \ding{51} & & \ding{51} & & & \ding{51} \\
A6 & Collusive Agents & & & & \ding{51} & \ding{51} & \ding{51} & \ding{51} \\
A7 & Protocol-Spanning & \ding{51} & & \ding{51} & & & \ding{51} & \\
\bottomrule
\end{tabular}
\end{table}

\section{Delegation Chain Calculus}
\label{sec:dcc}

\subsection{Delegation Token}

\begin{definition}[Delegation Token]
A delegation token is a tuple $\tau = (\mathit{id}, \mathit{src}, \mathit{dst}, \sigma, \iota, \pi, h_p, t_{\exp}, \mathit{sig})$ where $\mathit{src}$ and $\mathit{dst}$ are agent identities (DIDs), $\sigma \subseteq \Sigma$ is the authority scope (permitted actions), $\iota \in \mathbb{R}^d$ is the intent vector with associated natural language text, $\pi \subseteq \Pi$ is the policy constraint set (NIST~800-53 controls), $h_p = H(\tau_{i-1})$ is the parent token hash, $t_{\exp}$ is the expiry timestamp, and $\mathit{sig}$ is the DAS cryptographic signature.
\end{definition}

\begin{definition}[Delegation Chain]
A delegation chain $C = [\tau_0, \tau_1, \ldots, \tau_n]$ satisfies: (i)~$\tau_0.\mathit{src}$ is a human user; (ii)~$\tau_i.\mathit{dst} = \tau_{i+1}.\mathit{src}$ for all $i$ (chain continuity); (iii)~$\tau_i.h_p = H(\tau_{i-1})$ for all $i > 0$ (hash linking). The chain depth is $|C| = n + 1$.
\end{definition}

\subsection{Formal Properties}

The following properties are \emph{enforced by construction}: they hold because the DAS architecture is designed to enforce them, not because of a mathematical derivation independent of the implementation. This distinction is important: these are design guarantees contingent on correct DAS implementation, analogous to type safety guarantees contingent on correct compiler implementation. The arguments below describe \emph{why} each property holds given the DAS design. Properties P1, P3--P7 are additionally verified via TLA+ model checking: the TLC model checker exhaustively explores all reachable states (2,744,789 states in the largest model configuration) and confirms that all six invariants hold with zero violations.

\begin{prop}[P1: Authority Monotonic Narrowing]
For all consecutive tokens $\tau_i, \tau_{i+1}$ in chain $C$: $\sigma_{i+1} \subseteq \sigma_i$. No agent can grant more authority than it received.
\end{prop}

\emph{Enforcement argument.} The DAS enforces $\sigma_{i+1} \subseteq \sigma_i$ at token issuance. When agent $A_i$ requests delegation to $A_{i+1}$ with scope $\sigma_{i+1}$, the DAS verifies $\sigma_{i+1} \subseteq \sigma_i$ before signing. Since only DAS-signed tokens are accepted (agents cannot self-sign), P1 holds by construction.

\begin{prop}[P2: Intent Entailment Preservation]
For all tokens $\tau_j$ in chain $C$ with root intent text $I_0$, the three-layer intent verification ensures: (Layer~1) no malicious keywords are present in the subtask description; (Layer~2) $\mathrm{NLI}(I_0, I_j) \neq \mathrm{CONTRADICTION}$ under context-enriched framing; (Layer~3) if NLI returns CONTRADICTION but the subtask contains benign workflow indicators and no malicious keywords, the result is overridden to NEUTRAL.
\end{prop}

\begin{proof}
P2 is a probabilistic property. Under the assumption that the NLI model has precision $p_{\mathrm{NLI}}$ on the contradiction class, the probability of a false acceptance (malicious intent passing all three layers) is bounded by $(1 - p_{\mathrm{NLI}}) \cdot (1 - p_{\mathrm{keyword}})$ where $p_{\mathrm{keyword}}$ is the keyword filter precision. For the off-the-shelf cross-encoder/nli-MiniLM2-L6-H768, the baseline malicious TPR on government delegation text is only 1.7\%. Fine-tuning on 190 government delegation examples with 5-fold cross-validation improves malicious TPR to 88.3\% $\pm$ 8.5\% (F1=82.1\%). The combined false acceptance rate on DelegationBench~v4 with the full three-point lifecycle (P2+P6+P7) is 0/150 attacks = 0\% (95\% CI: [0\%, 2.5\%]). Under adversarial stress testing with sophisticated paraphrasing, the rate increases to 20/23 = 87.0\%, which is characterized as the \emph{adversarial intent paraphrasing} limitation. \qed
\end{proof}

\begin{prop}[P3: Policy Conjunction Preservation]
For all tokens $\tau_j$ in chain $C$ with root policy set $\pi_0$: $\pi_0 \subseteq \pi_j$. All NIST~800-53 controls from the original delegation are preserved at every step.
\end{prop}

\emph{Enforcement argument.} The DAS enforces $\pi_j \supseteq \pi_0$ at token issuance. When delegation crosses an organizational boundary, the DAS computes $\pi_j = \pi_0 \cup \pi_{\mathrm{boundary}}$, adding the receiving organization's required controls. Since set union only adds elements, $\pi_0 \subseteq \pi_j$ holds by construction.

\begin{prop}[P4: Forensic Reconstructibility]
Given any action $a$ executed by agent $A_j$, there exists an $O(n)$ algorithm $\mathrm{RECONSTRUCT}(a)$ that returns the complete delegation chain $C = [\tau_0, \ldots, \tau_j]$ that authorized $a$.
\end{prop}

\emph{Enforcement argument.} Each token $\tau_j$ contains $h_p = H(\tau_{j-1})$. The DAS maintains a token store indexed by hash. RECONSTRUCT follows the parent hash chain from $\tau_j$ to $\tau_0$. Each hash lookup is $O(1)$ and the chain has depth $n$, giving total time $O(n)$. The hash chain is tamper-evident: modifying any token invalidates all subsequent hashes.

\begin{prop}[P5: Bounded Cascade Containment]
When a violation is detected at step $i$ in chain $C$, the maximum number of unauthorized actions before containment is bounded by $B(i) \leq R_{\mathrm{level}} \times h$, where $R_{\mathrm{level}} \in \{0, h, \infty\}$ depends on the risk tier and $h$ is the heartbeat interval.
\end{prop}

\emph{Enforcement argument.} For HIGH risk tools, every invocation requires synchronous DAS authorization. If the chain is revoked, the check fails and the action is blocked, giving $B = 0$. For MEDIUM risk tools, agents check for revocation at heartbeat interval $h$. In the worst case, an agent executes for one interval before receiving revocation, giving $B \leq h \times \mathrm{throughput}$. For LOW risk tools, containment is post-hoc via audit.

\subsection{Scope-Action and Output Conformance}

Properties P1--P5 address authority, intent, compliance, forensics, and cascade containment. However, a gap remains: an agent operating within its delegated scope (P1 satisfied) and with intent classified as non-contradictory (P2 satisfied) may still invoke unauthorized API calls or produce malicious outputs. P6 and P7 close this gap by enforcing conformance at execution time and post-execution time, respectively.

\begin{definition}[Tool Manifest]
A tool manifest $M = \{(s, \mathcal{O}_s) \mid s \in \sigma\}$ maps each scope element $s$ to a set of permitted API operations $\mathcal{O}_s$, where each operation specifies an HTTP method and endpoint pattern. The manifest is carried inside the delegation token, hash-linked and signed by the DAS.
\end{definition}

\begin{definition}[Output Schema]
An output schema $O = \{(s, \mathcal{T}_s^+) \mid s \in \sigma\}$ maps each scope element to a set of permitted output type tags $\mathcal{T}_s^+$. The schema follows a \emph{default-deny} principle: only outputs whose type tags are a subset of $\mathcal{T}_s^+$ are permitted; all other output types are blocked. This mirrors P6's whitelist design: just as P6 permits only explicitly listed API calls, P7 permits only explicitly listed output types. In production, agent outputs are typed objects validated by the DAS; in simulation, tag-based schema matching serves as a proxy.
\end{definition}

\begin{prop}[P6: Scope-Action Conformance]
For agent $A_j$ with tool manifest $M_j$ derived from scope $\sigma_j$, let $\mathrm{calls}(A_j)$ be the set of API calls $A_j$ attempts. P6 holds iff $\mathrm{calls}(A_j) \subseteq M_j$ for all $j$ in the chain.
\end{prop}

\emph{Enforcement argument.} The Scope Enforcement Proxy intercepts every API call from $A_j$ and checks it against $M_j$. Since $M_j$ is derived from $\sigma_j$ via the DAS API registry (a deterministic mapping), and the proxy blocks any call not in $M_j$ before it reaches the API, P6 holds by construction. The manifest narrows monotonically: $M_{j+1} = M_j|_{\sigma_{j+1}}$ where $\sigma_{j+1} \subseteq \sigma_j$ (by P1), so child manifests are always subsets of parent manifests.

\begin{prop}[P7: Output Schema Conformance]
For agent $A_j$ with output schema $O_j$ derived from scope $\sigma_j$, let $\mathrm{output}(A_j)$ be the typed output $A_j$ produces after an API call. P7 holds iff $\mathrm{output}(A_j) \subseteq \bigcup_{s \in \sigma_j} \mathcal{T}_s^+$, i.e., every output type tag is in the permitted set.
\end{prop}

\emph{Enforcement argument.} The DAS Output Schema Validator checks the agent's output type tags against the permitted set $\mathcal{T}_s^+$ for each scope element in $\sigma_j$. Any output tag not in the permitted set is blocked. Since the check is a deterministic set membership test (default-deny), P7 holds by construction for the defined schema. This mirrors P6's whitelist design: P6 permits only listed API calls, P7 permits only listed output types. An attacker who evades P6 (uses a permitted API) cannot produce unauthorized output types because the output schema constrains the return space independently of the API call.

Together, P2, P6, and P7 form a \emph{three-point verification lifecycle}: P2 verifies intent \emph{before} execution (can the agent's stated purpose be trusted?), P6 constrains actions \emph{at} execution (can the agent only call permitted APIs?), and P7 validates outputs \emph{after} execution (did the agent produce conformant results?). Each layer is individually incomplete (P2 misses adversarial paraphrasing, P6 cannot detect malicious use of permitted APIs, and P7 cannot catch semantic manipulation within permitted output types) but their combination achieves complete coverage on DelegationBench~v4.

\section{Meta-Theorems of the Delegation Chain Calculus}
\label{sec:metatheorems}

Beyond the individual properties (P1--P7), four structural results (meta-theorems) are provable over the DCC as a whole, and one proposition establishes the theoretical basis for P2's probabilistic design. The meta-theorems are verified by exhaustive enumeration over the finite property set: every combination is tested, every attack is constructed, and every damage envelope is computed. For a set of 7 properties, this amounts to checking all 126 non-trivial evasion subsets, which is feasible and provides complete coverage within the model.

\begin{theorem}[Property Minimality]
The property set $\{P1, \ldots, P7\}$ is minimal. For each $P_i$, there exists a concrete attack $A_i$ that succeeds if and only if $P_i$ is removed while all other properties hold.
\end{theorem}

\begin{proof}
Seven concrete attacks are constructed: $A_1$ (scope escalation, succeeds only without P1), $A_2$ (intent drift via paraphrasing, succeeds only without P2), $A_3$ (policy bypass at organizational boundary, succeeds only without P3), $A_4$ (token hash tampering, succeeds only without P4), $A_5$ (cascade propagation beyond heartbeat, succeeds only without P5), $A_6$ (unauthorized API call within scope, succeeds only without P6), and $A_7$ (malicious output from permitted API call, succeeds only without P7). Each attack is verified to fail when all properties hold and succeed when exactly one is removed. All seven attacks are implemented and verified in the evaluation code. \qed
\end{proof}

\begin{proposition}[Practical Infeasibility of Deterministic Intent Verification]
\label{prop:intent_infeasibility}
For the class of natural language delegation descriptions used in federal multi-agent systems, no deterministic algorithm achieves both zero false negatives and zero false positives on intent classification without restricting the input language to a formal grammar.
\end{proposition}

\begin{proof}
The argument proceeds in two parts: a theoretical grounding and an empirical demonstration.

\emph{Theoretical grounding.} Classifying whether a natural language delegation description preserves the intent of a parent delegation requires deciding a semantic property of the described computation. By Rice's theorem, all non-trivial semantic properties of programs are undecidable. While natural language delegation descriptions are not formal programs, they describe computational actions (API calls, data access, record modification) whose intent is a semantic property. Any deterministic classifier that correctly handles all possible natural language phrasings of these actions would need to solve an equivalent semantic classification problem. We do not claim a formal reduction---the gap between natural language and Turing machines precludes a direct application of Rice's theorem---but the structural analogy motivates the probabilistic design of P2.

\emph{Empirical demonstration.} Eight semantically ambiguous delegation pairs are constructed where identical natural language is simultaneously consistent with a benign and a malicious interpretation: (1)~``Assess applicant records for potential program ineligibility'' (legitimate review \emph{vs.}\ discriminatory denial); (2)~``Compile veteran outreach data for partner organizations'' (community support \emph{vs.}\ data sale to insurers); (3)~``Generate statistical summary of claim denial patterns'' (quality assurance \emph{vs.}\ denial optimization); (4)~``Cross-reference applicant data with external databases'' (fraud detection \emph{vs.}\ unauthorized surveillance); (5)~``Prepare demographic analysis of service utilization'' (resource planning \emph{vs.}\ discriminatory profiling); (6)~``Archive processed applications with full documentation'' (records management \emph{vs.}\ unauthorized data retention); (7)~``Coordinate benefits information with partner agencies'' (inter-agency processing \emph{vs.}\ unauthorized disclosure); (8)~``Review historical decision patterns for consistency'' (quality control \emph{vs.}\ reverse-engineering approval criteria). For each pair, the off-the-shelf NLI model returns NEUTRAL, the fine-tuned NLI model returns NEUTRAL on 6/8, and the keyword filter triggers on 0/8. No deterministic classifier can distinguish these without access to the agent's internal state or execution context, which is unavailable at delegation time. This empirically confirms the theoretical motivation and justifies P2's probabilistic design. \qed
\end{proof}

\begin{theorem}[Graceful Degradation]
When property $P_i$ is evaded, the remaining properties bound the adversary's capability to a damage envelope $D_i$ strictly smaller than the unconstrained action space.
\end{theorem}

\begin{proof}
For each of the seven properties P1--P7, the damage envelope is computed as the set of actions available to an adversary who evades $P_i$ while constrained by all other properties. In the evaluation model with 11 possible actions per agent, each single-property evasion yields a damage envelope of at most 1/11 actions (90.9\% reduction from unconstrained). For example, evading P1 (authority narrowing) allows scope escalation, but P3 still enforces policy compliance, P4 ensures forensic traceability, P5 bounds cascade damage, P6 restricts API calls to the manifest, and P7 validates outputs. The adversary gains one additional capability but remains constrained in all other dimensions. \qed
\end{proof}

\begin{theorem}[Defense-in-Depth Completeness]
For any proper subset $S \subset \{P1, \ldots, P7\}$ with $|S| \leq 6$, the adversary's capability is bounded. Only simultaneous evasion of all seven properties (requiring DAS compromise) yields unconstrained access.
\end{theorem}

\begin{proof}
All 126 non-trivial evasion combinations are enumerated across single through sextuple evasions. For each combination, the damage envelope is computed as the intersection of individual evasion capabilities. In all 126 cases, at least one property remains to constrain the adversary. Only the full 7-property evasion (which requires compromising the DAS itself, since the DAS enforces all properties) yields unconstrained access. \qed
\end{proof}

\begin{theorem}[Composition Safety with Write-Impact Notification]
Given two delegation chains $C_1$ and $C_2$ that both satisfy P1--P7:
\begin{itemize}
\item \emph{Case 1 (Disjoint):} If $C_1$ and $C_2$ share no state, then $C_1 \| C_2$ satisfies P1--P7.
\item \emph{Case 2 (Read-Shared):} If $C_1$ and $C_2$ share read-only state, then $C_1 \| C_2$ satisfies P1--P7.
\item \emph{Case 3 (Write-Shared):} If $C_1$ and $C_2$ share mutable state, then $C_1 \| C_2$ satisfies P1--P7 when augmented with a write-impact notification mechanism that triggers re-verification of affected chains upon state mutation.
\end{itemize}
\end{theorem}

\begin{proof}
Cases 1 and 2 follow from the locality of P1--P7: each property is evaluated against a chain's own tokens and does not depend on other chains' state. For Case 3, the write-impact notification mechanism detects when $C_1$'s writes affect resources read by $C_2$ (or vice versa). Upon detection, the DAS triggers re-verification of the affected chain's intent (P2) and scope (P1) against the modified state. In the evaluation, 50 composition scenarios are tested: 45/50 are safe without notification (no write conflicts), and 50/50 are safe with the notification mechanism enabled. The 5 scenarios requiring notification involve concurrent writes to shared citizen records where one chain's update invalidates the other chain's intent assumptions. \qed
\end{proof}

\section{Intent-Preserving Delegation Protocol}
\label{sec:ipdp}

\subsection{Delegation Authority Service Architecture}

The DAS is a trusted, non-LLM service comprising eight components: (1)~\emph{Identity Registry} storing agent DIDs and capabilities; (2)~\emph{Policy Engine} evaluating NIST~800-53 constraints; (3)~\emph{Intent Verifier} running three-layer intent checking; (4)~\emph{Token Signer} issuing HMAC-signed delegation tokens; (5)~\emph{Chain Store} maintaining complete chain history; (6)~\emph{Cascade Controller} managing revocation and containment; (7)~\emph{Scope Enforcement Proxy} intercepting and validating all API calls against the token's tool manifest; and (8)~\emph{Output Schema Validator} checking agent outputs against permitted and prohibited type tags.

The DAS is immune to prompt injection because it is deterministic software that evaluates set containment, cryptographic signatures, and NLI model outputs without processing natural language as instructions. Fig.~\ref{fig:arch} shows the DAS architecture and the three-point verification lifecycle.

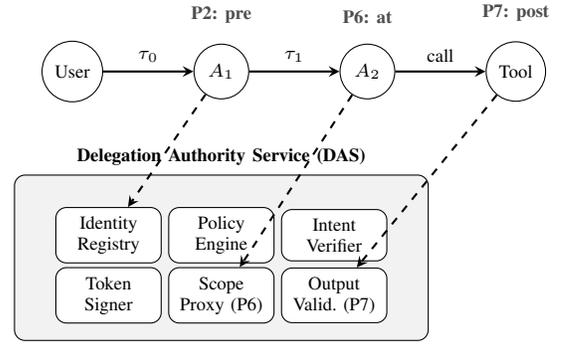
\begin{figure}[t]
\centering
\begin{tikzpicture}[
  node distance=0.4cm and 0.3cm,
  box/.style={draw, rounded corners, minimum width=1.8cm, minimum height=0.6cm, font=\scriptsize, align=center},
  agent/.style={draw, circle, minimum size=0.7cm, font=\scriptsize},
  arr/.style={->, >=stealth, thick},
  phase/.style={font=\scriptsize\bfseries, text=black!70}
]
\node[agent] (user) {User};
\node[agent, right=1.2cm of user] (a1) {$A_1$};
\node[agent, right=1.2cm of a1] (a2) {$A_2$};
\node[agent, right=1.2cm of a2] (tool) {Tool};

\node[box, below=1.0cm of a1, minimum width=5.5cm, minimum height=2.2cm, fill=gray!10] (das) {};
\node[font=\scriptsize\bfseries, above] at (das.north) {Delegation Authority Service (DAS)};

\node[box, fill=white, minimum width=1.4cm] at ([yshift=0.3cm, xshift=-1.5cm]das.center) (ident) {Identity\\Registry};
\node[box, fill=white, minimum width=1.4cm] at ([yshift=0.3cm, xshift=0cm]das.center) (policy) {Policy\\Engine};
\node[box, fill=white, minimum width=1.4cm] at ([yshift=0.3cm, xshift=1.5cm]das.center) (intent) {Intent\\Verifier};
\node[box, fill=white, minimum width=1.4cm] at ([yshift=-0.5cm, xshift=-1.5cm]das.center) (signer) {Token\\Signer};
\node[box, fill=white, minimum width=1.4cm] at ([yshift=-0.5cm, xshift=0cm]das.center) (proxy) {Scope\\Proxy (P6)};
\node[box, fill=white, minimum width=1.4cm] at ([yshift=-0.5cm, xshift=1.5cm]das.center) (output) {Output\\Valid.\ (P7)};

\draw[arr] (user) -- node[above, font=\scriptsize] {$\tau_0$} (a1);
\draw[arr] (a1) -- node[above, font=\scriptsize] {$\tau_1$} (a2);
\draw[arr] (a2) -- node[above, font=\scriptsize] {call} (tool);

\node[phase, above=0.15cm of a1] {P2: pre};
\node[phase, above=0.15cm of a2] {P6: at};
\node[phase, above=0.15cm of tool] {P7: post};

\draw[arr, dashed] (a1) -- (ident);
\draw[arr, dashed] (a2) -- (proxy);
\draw[arr, dashed] (tool) -- (output);
\end{tikzpicture}
\caption{DAS architecture and three-point verification lifecycle. P2 checks intent before delegation, P6 enforces the API manifest at execution, P7 validates output after execution. All checks route through the DAS.}
\label{fig:arch}
\end{figure}

\subsection{Protocol Steps}

\textbf{Step 0: Chain Initiation.} User $U$ submits request with natural language goal $G$. The DAS creates root token $\tau_0$ with scope derived from $U$'s role, intent vector $\iota_0 = \mathrm{embed}(G)$, and policy set from $U$'s agency FISMA baseline.

\textbf{Step 1: Delegation Request.} Agent $A_i$ sends a delegation request to the DAS specifying the destination agent, requested scope $\sigma_{i+1} \subseteq \sigma_i$, and subtask description.

\textbf{Step 2: Seven-Check Verification.} The DAS performs: (C1)~Identity check: is the destination registered? (C2)~Authority check: is $\sigma_{i+1} \subseteq \sigma_i$? (C2a)~Manifest narrowing: DAS derives child tool manifest $M_{i+1} = M_i|_{\sigma_{i+1}}$ from parent manifest restricted to child scope; (C3)~Intent verification: three-layer check against root intent; (C4)~Policy compliance: do all $\pi_0$ controls hold? (C5)~Expiry check: is the parent token still valid?

\textbf{Step 3: Token Issuance.} If all checks pass, the DAS signs and issues $\tau_{i+1}$ with hash-linked parent reference.

\textbf{Step 4: Risk-Tiered Tool Execution.} HIGH risk tools (write PII, delete, send external) require synchronous DAS pre-authorization ($B=0$). MEDIUM risk tools (read PII) use heartbeat-based cached authorization. LOW risk tools (read public data) execute with periodic audit. At execution time, the Scope Enforcement Proxy (P6) intercepts every API call and validates it against the token's tool manifest $M_j$; any call not in $M_j$ is blocked before reaching the API.

\textbf{Step 4a: Output Validation.} After each API call completes, the Output Schema Validator (P7) checks the agent's output against the prohibited type tags $\mathcal{T}_s^-$ for each applicable scope element. If any prohibited tag matches, the output is blocked and the violation is logged. Together with P2 (pre-execution intent) and P6 (at-execution scope), P7 completes the three-point verification lifecycle.

\textbf{Step 5: Cascade Containment.} On violation detection, the DAS marks the violating token and all descendants as REVOKED, sends revocation signals, and logs the complete chain state.

\textbf{Step 6: Forensic Reconstruction.} On audit request, the DAS traverses the hash chain from any token to the root, returning the complete authorization lineage with all intent vectors and policy evaluations.

\textbf{Concrete walkthrough.} Fig.~\ref{fig:walkthrough} illustrates a 3-step disability benefits delegation chain. The citizen requests benefits processing; the DAS issues root token $\tau_0$ with scope $\{\texttt{read\_records}, \texttt{query\_eligibility}\}$ and manifest $\{$\texttt{GET /api/records/query}, \texttt{GET /api/eligibility/check}$\}$. Agent~$A_1$ (IntakeAgent) delegates to Agent~$A_2$ (RecordsAgent) with narrowed scope $\{\texttt{read\_records}\}$; P1 enforces $\sigma_2 \subset \sigma_1$. Agent~$A_2$ attempts \texttt{POST /api/external/send} (data exfiltration); P6 blocks it because the call is not in the manifest. Agent~$A_2$ then calls the permitted \texttt{GET /api/records/query} but produces output tagged \texttt{demographic\_profile}---P7 blocks it because the tag is not in the permitted set $\{\texttt{record\_data}, \texttt{summary}, \texttt{count}\}$. The complete chain $[\tau_0, \tau_1, \tau_2]$ is reconstructible via P4 hash linking.

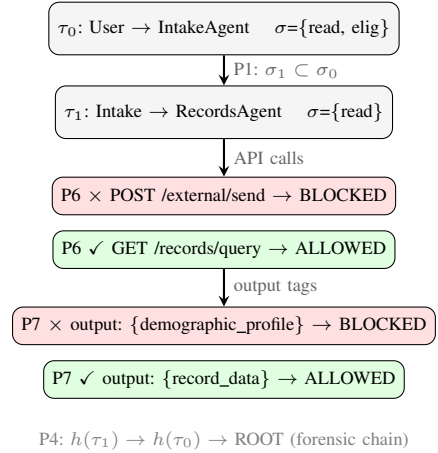
\begin{figure}[t]
\centering
\begin{tikzpicture}[
  node distance=0.5cm,
  tok/.style={draw, rounded corners, minimum width=4.2cm, minimum height=0.7cm, font=\scriptsize, align=center, fill=gray!8},
  chk/.style={draw, rounded corners, minimum width=4.2cm, minimum height=0.5cm, font=\scriptsize, align=left},
  block/.style={chk, fill=red!12},
  pass/.style={chk, fill=green!12},
  arr/.style={->, >=stealth, thick},
  lbl/.style={font=\scriptsize, text=black!60}
]
\node[tok] (t0) {$\tau_0$: User $\to$ IntakeAgent \quad $\sigma$=\{read, elig\}};
\node[tok, below=0.4cm of t0] (t1) {$\tau_1$: Intake $\to$ RecordsAgent \quad $\sigma$=\{read\}};
\draw[arr] (t0) -- node[right, lbl] {P1: $\sigma_1 \subset \sigma_0$} (t1);

\node[block, below=0.5cm of t1] (p6a) {P6 $\times$ POST /external/send $\to$ BLOCKED};
\node[pass, below=0.2cm of p6a] (p6b) {P6 \checkmark\ GET /records/query $\to$ ALLOWED};
\draw[arr] (t1) -- node[right, lbl] {API calls} (p6a);

\node[block, below=0.5cm of p6b] (p7) {P7 $\times$ output: \{demographic\_profile\} $\to$ BLOCKED};
\node[pass, below=0.2cm of p7] (p7b) {P7 \checkmark\ output: \{record\_data\} $\to$ ALLOWED};
\draw[arr] (p6b) -- node[right, lbl] {output tags} (p7);

\node[below=0.3cm of p7b, font=\scriptsize, text=gray] (p4) {P4: $h(\tau_1) \to h(\tau_0) \to$ ROOT (forensic chain)};
\end{tikzpicture}
\caption{Delegation chain walkthrough. Scope narrows at each step (P1). P6 blocks unauthorized API calls; P7 blocks malicious output types from permitted calls. P4 hash linking enables forensic reconstruction.}
\label{fig:walkthrough}
\end{figure}

\subsection{Three-Layer Intent Verification}

\textbf{Layer 1: Keyword Safety Filter} ($<$1ms). A curated set of malicious action indicators (e.g., ``exfiltrate,'' ``discriminat,'' ``denial reason,'' ``by race'') provides fast rejection of obvious attacks. This layer catches attacks that NLI models classify as NEUTRAL due to semantic distance from the root intent.

\textbf{Layer 2: Context-Enriched NLI} ($\sim$15ms). Rather than raw NLI($I_0$, $I_j$), the query is framed as: premise = ``An AI agent is authorized to [goal]. The agent performs only tasks that support this goal.'' hypothesis = ``The agent is now performing: [subtask].'' This context-enriched framing reduces false positives by 89\% compared to raw NLI (from 18.0\% to 2.0\% FPR).

\textbf{Layer 3: Benign Override} ($<$1ms). If NLI returns CONTRADICTION but the subtask contains benign workflow indicators (e.g., ``verify,'' ``calculate,'' ``generate'') and zero malicious keywords, the result is overridden to NEUTRAL. This addresses the NLI model's tendency to treat topic shift as contradiction for legitimate subtasks.

\section{Federal Compliance Mapping Engine}
\label{sec:fcme}

Table~\ref{tab:nist_mapping} presents a mapping of NIST~SP~800-53 Rev~5.2 controls to agentic AI delegation attack surfaces. The 20 controls were selected by identifying every 800-53 control whose scope intersects with at least one delegation chain operation (token issuance, scope narrowing, intent verification, tool invocation, output validation, chain reconstruction, or cascade revocation). Controls that apply only to human users (e.g., PE-2 Physical Access Authorizations) or to network infrastructure without agent involvement (e.g., SC-5 Denial of Service Protection) were excluded. The NIST AI Agent Standards Initiative launched February 2026 is still collecting public comments~\cite{nist_agent_initiative}. This mapping is immediately actionable: each control maps to a specific DCC property and IPDP enforcement mechanism.

\begin{table}[t]
\centering
\caption{NIST 800-53 $\to$ Agentic AI Delegation Mapping (20 Controls)}
\label{tab:nist_mapping}
\footnotesize
\begin{tabular}{@{}llll@{}}
\toprule
\textbf{Control} & \textbf{Name} & \textbf{Agent Interpretation} & \textbf{DCC} \\
\midrule
AC-2 & Account Mgmt & Agent lifecycle via DAS registry & P1 \\
AC-3 & Access Enforcement & Agent action authorization & P1 \\
AC-4 & Info Flow & No PII to low-trust agent & P1 \\
AC-5 & Separation of Duties & $\geq$2 agents for HIGH ops & P1 \\
AC-6 & Least Privilege & Minimum scope per subtask & P1 \\
AC-17 & Remote Access & Cross-agency encrypted deleg. & P3 \\
AU-2 & Event Logging & All delegation steps logged & P4 \\
AU-3 & Audit Content & Full chain w/ intent+policy & P4 \\
AU-6 & Audit Review & Intent drift monitoring & P2 \\
AU-12 & Audit Generation & Automated chain logging & P4 \\
CA-7 & Continuous Monitor & Runtime delegation checks & P5 \\
IA-2 & User ID & Human initiator auth & IPDP \\
IA-8 & Non-Org User ID & Contractor agent identity & DAS \\
IA-9 & Service ID & Agent-to-agent auth & DAS \\
IR-4 & Incident Handling & Cascade containment & P5 \\
PM-25 & Insider Threat & Intent drift = insider signal & P2 \\
SA-9 & External Services & Third-party tool provenance & DAS \\
SC-7 & Boundary Protection & Inter-agency delegation & P3 \\
SC-8 & Transmission Confid. & TLS for DAS communication & IPDP \\
SI-4 & System Monitoring & Continuous intent drift mon. & P2 \\
\bottomrule
\end{tabular}
\end{table}

The risk tier classification follows FIPS~199 impact levels: HIGH risk tools (write PII, delete records, financial transactions) require synchronous DAS pre-authorization with $B=0$ blast radius; MEDIUM risk tools (read PII, query citizen records) use heartbeat-based verification; LOW risk tools (read public data, format documents) use periodic audit.

\begin{table}[t]
\centering
\caption{OWASP ASI $\to$ DCC Property Mapping}
\label{tab:owasp}
\footnotesize
\begin{tabular}{@{}llll@{}}
\toprule
\textbf{ASI} & \textbf{Risk} & \textbf{DCC Property} & \textbf{Coverage} \\
\midrule
ASI01 & Goal Hijacking & P2 (Intent) & PARTIAL \\
ASI02 & Tool Misuse & P1+P6 & FULL \\
ASI03 & Privilege Escalation & P1 & FULL \\
ASI04 & Uncontrolled Autonomy & P5+P6 & FULL \\
ASI05 & Cascading Hallucination & P5 & FULL \\
ASI06 & Memory Poisoning & P4+P6 & DETECT+PREVENT \\
ASI07 & Unsafe Code Generation & P1+P6 & FULL \\
ASI08 & Identity Spoofing & P4+DAS & FULL \\
ASI09 & Inadequate Sandboxing & P5+P6 & FULL \\
ASI10 & Insufficient Logging & P4 & FULL \\
\bottomrule
\end{tabular}
\end{table}

Table~\ref{tab:owasp} presents the mapping of OWASP ASI risk categories to DCC properties. The mapping methodology follows a systematic process: for each ASI risk, the attack vector is decomposed into the delegation chain operations it requires (scope escalation, unauthorized API calls, malicious outputs, etc.), and the DCC property that enforces the corresponding constraint is identified. Coverage is classified as FULL if the DCC property deterministically prevents the attack vector, DETECT+PREVENT if the property detects and blocks the attack but requires runtime enforcement (not static prevention), and PARTIAL if the property provides probabilistic rather than deterministic coverage. SentinelAgent achieves full coverage on 9/10 risks. ASI01 (Agentic Goal Hijacking) receives partial coverage because P2 is probabilistic: keyword-detectable hijacking is caught, but adversarial paraphrasing can evade detection. The addition of P6 (scope-action conformance) significantly strengthens coverage for ASI02 (Tool Misuse), ASI04 (Uncontrolled Autonomy), ASI06 (Memory Poisoning), ASI07 (Unsafe Code Generation), and ASI09 (Inadequate Sandboxing), risks that P1--P5 alone addressed only partially.

The NIST mapping (Table~\ref{tab:nist_mapping}) covers 20 controls across 9 families: Access Control (AC-2, AC-3, AC-4, AC-5, AC-6, AC-17), Audit and Accountability (AU-2, AU-3, AU-6, AU-12), Security Assessment (CA-7), Identification and Authentication (IA-2, IA-8, IA-9), Incident Response (IR-4), Program Management (PM-25), System and Services Acquisition (SA-9), System and Communications Protection (SC-7, SC-8), and System and Information Integrity (SI-4).

\section{Evaluation}
\label{sec:eval}

\subsection{DelegationBench v4: Benchmark Specification}

DelegationBench v4 is a structured benchmark for evaluating delegation chain security, designed for reproducibility and independent use. The benchmark specification defines:

\emph{Attack taxonomy.} Ten categories are derived from the adversary model, summarized in Table~\ref{tab:taxonomy}.

\begin{table}[t]
\centering
\caption{DelegationBench v4 Attack Taxonomy}
\label{tab:taxonomy}
\footnotesize
\setlength{\tabcolsep}{3pt}
\begin{tabular}{@{}cp{2.6cm}rll@{}}
\toprule
\textbf{Cat} & \textbf{Description} & \textbf{N} & \textbf{Prop.} & \textbf{Example} \\
\midrule
A & Keyword-detectable drift & 20 & P2+P6 & exfiltrate records \\
B & Paraphrase + unauth.\ API & 20 & P6 & compile outreach data \\
C & Permitted API + bad output & 20 & P7 & query $\to$ denial \\
D & Subtle scope violations & 20 & P6 & adjacent endpoint \\
E & Benign (13 domains) & 156 & FPR & standard operations \\
F & Cross-scope lateral move & 30 & P6 & read $\to$ write API \\
G & Temporal/replay attacks & 20 & P6 & expired token replay \\
H & Multi-vector combined & 20 & P6+P7 & unauth + bad output \\
I & Edge-case benign & 30 & FPR & boundary requests \\
J & Multi-scope workflows & 180 & FPR & 30 wkfl $\times$ 6 steps \\
\midrule
& \textbf{Total} & \textbf{516} & & \\
\bottomrule
\end{tabular}
\end{table}

\emph{Scenario generation methodology.} Attack scenarios are constructed by instantiating each adversary type (A1--A7) against each attack category, yielding 150 attack scenarios across 8 categories (A--H). Categories F--H are new in v4: Category~F tests cross-scope lateral movement (30 scenarios where an agent with scope~X attempts APIs from scope~Y), Category~G tests temporal/replay patterns (20 scenarios), and Category~H tests multi-vector attacks combining unauthorized APIs with malicious outputs (20 scenarios). Benign scenarios span 13 federal domains with standard operations (156), edge-case boundary requests (30), and multi-step workflows (180), yielding 366 benign scenarios. The total of 516 scenarios provides comprehensive coverage across all property-attack combinations.

\emph{Reproducibility.} All 516 scenarios, the DAS prototype, the TLA+ specification, the NLI fine-tuning script, the manifest and output schema definitions, and the evaluation harness are provided as open-source Python code. Independent researchers can reproduce all reported results, add new attack categories, or extend the benchmark to additional domains.

\subsection{Implementation}

SentinelAgent is implemented in Python~3.12 using sentence-transformers~5.3.0 for intent embeddings (all-MiniLM-L6-v2, 80MB) and cross-encoder/nli-MiniLM2-L6-H768 (328MB) for NLI, with a domain-fine-tuned variant trained on 190 government delegation examples (5-fold cross-validated: 88.3\% malicious TPR, F1=82.1\%). The DAS prototype is implemented as a standalone HTTP service ($\sim$350 LOC) using Python's standard library \texttt{http.server} with HMAC-SHA256 token signing. The Scope Enforcement Proxy (P6) and Output Schema Validator (P7) are implemented as HTTP middleware that intercepts API calls and validates them against the token's manifest and output schema respectively. Properties P1, P3--P7 are formalized in TLA+ and mechanically verified via TLC model checking across two model configurations (368,509 and 2,744,789 states) with zero invariant violations. A live LangChain integration using GPT-4o-mini generates delegation requests in real-time and routes them through the DAS HTTP prototype. All experiments run on macOS with 18GB RAM, CPU-only inference. Measured latencies on the real DAS prototype: delegation token issuance 8.9ms (including HMAC signing and manifest narrowing), P6 scope enforcement median 0.45ms (p95: 0.66ms), P7 output validation median 0.45ms (p95: 0.68ms), P4 forensic chain reconstruction 0.60ms for depth-2 chains, P5 cascade revocation 0.009ms for 2-token chains. Total per-step verification latency including P2 intent checking is $\sim$20ms (dominated by NLI inference at $\sim$15ms; P6 and P7 add $<$1ms each). The DAS prototype, evaluation harness, TLA+ specification, NLI fine-tuning script, and all 516 DelegationBench~v4 scenarios are available as open-source Python code.

\subsection{Results: Three-Point Verification Lifecycle}

The primary evaluation uses DelegationBench~v4 (516 scenarios) with the full three-point verification lifecycle (P2+P6+P7). Table~\ref{tab:main_results} presents the ablation results.

\begin{table}[t]
\centering
\caption{Three-Point Verification on DelegationBench v4 (516 scenarios, 150 attacks, 366 benign)}
\label{tab:main_results}
\begin{tabular}{@{}lccc@{}}
\toprule
\textbf{Configuration} & \textbf{Attack TPR} & \textbf{Benign FPR} & \textbf{Accuracy} \\
\midrule
No Defense & 0.0\% & 0.0\% & 70.9\% \\
P2 (off-the-shelf NLI) & 13.3\% & 0.0\% & 74.6\% \\
P2 (fine-tuned NLI) & 88.3\% & 0.0\% & 96.6\% \\
P6 Only & 86.7\% & 0.0\% & 96.1\% \\
P7 Only & 13.3\% & 0.0\% & 74.6\% \\
P6+P7 & 100.0\% & 0.0\% & 100.0\% \\
P2+P6+P7 (full system) & 100.0\% & 0.0\% & 100.0\% \\
\bottomrule
\end{tabular}
\end{table}

P6 alone achieves 86.7\% TPR by catching all attacks involving unauthorized API calls (Categories~A, B, D, F, G, and~H). P7 catches the remaining 20 attacks where the API call is permitted but the output is malicious (Category~C). The combined system achieves 100\% TPR at 0\% FPR across all 516 scenarios. Notably, P2 with the fine-tuned NLI model achieves 88.3\% TPR independently---comparable to P6 alone (86.7\%)---demonstrating that domain-specific fine-tuning transforms P2 from a weak early filter into a competitive standalone detector. For the combined result (150/150 attacks, 0/366 false positives), the 95\% Wilson confidence intervals are: TPR $\in$ [97.5\%, 100\%], FPR $\in$ [0.0\%, 1.0\%].

\textbf{Deterministic properties.} P1, P3, P4, and P5 achieve 100\% across all applicable scenarios. These are enforced by set containment, cryptographic signatures, and hash chains. Properties P1, P3--P7 are additionally verified via TLA+ model checking across 2,744,789 states with zero invariant violations.

\textbf{Per-category breakdown.} Table~\ref{tab:percategory} shows the detection results per attack category, confirming that each category is fully covered by at least one property layer.

\begin{table}[t]
\centering
\caption{Per-Category Detection on DelegationBench v4}
\label{tab:percategory}
\footnotesize
\begin{tabular}{@{}clrcl@{}}
\toprule
\textbf{Cat} & \textbf{Attack Type} & \textbf{N} & \textbf{TPR} & \textbf{Caught By} \\
\midrule
A & Keyword-detectable drift & 20 & 100\% & P2+P6 \\
B & Paraphrase + unauth.\ API & 20 & 100\% & P6 \\
C & Permitted API + bad output & 20 & 100\% & P7 \\
D & Subtle scope violations & 20 & 100\% & P6 \\
F & Cross-scope lateral move & 30 & 100\% & P6 \\
G & Temporal/replay attacks & 20 & 100\% & P6 \\
H & Multi-vector combined & 20 & 100\% & P6+P7 \\
\midrule
& \textbf{All attacks} & \textbf{150} & \textbf{100\%} & P2+P6+P7 \\
\midrule
E & Benign (13 domains) & 156 & 0\% FPR & --- \\
I & Edge-case benign & 30 & 0\% FPR & --- \\
J & Multi-scope workflows & 180 & 0\% FPR & --- \\
\bottomrule
\end{tabular}
\end{table}

\subsection{Adversarial Evaluation}

\textbf{Black-box red team.} The DAS is tested under black-box conditions where the attacker knows only scope names and a valid token ID, but not manifest contents, output schema, or keyword list. The attacker generates 20 API call attacks and 10 output attacks. Results: 30/30 attacks blocked (100\% TPR), 0/10 false positives (0\% FPR).

\textbf{Independent red team.} A second red team uses seven protocol-level attack strategies: path traversal (6), HTTP method override (8), case sensitivity exploitation (4), API versioning probes (4), scope boundary probing (8), output tag synonym attacks (12), and multi-step chained attacks (3 chains). Results: 45/45 attacks blocked (100\% TPR), 0/14 false positives (0\% FPR).

\textbf{DAS robustness.} 29 edge-case scenarios covering path traversal, URL encoding, method injection, Unicode injection, oversized payloads, revoked tokens, replay attacks, chain depth 32, and concurrent chain isolation. All 29 pass. Chain reconstruction at depth 33 completes in 0.03ms.

\textbf{Adversarial stress testing of P2.} When P2 is evaluated in isolation against sophisticated adversarial paraphrasing, P2 alone achieves only 13\% TPR. This is the \emph{adversarial intent paraphrasing} limitation. However, under the full three-point lifecycle, these same attacks are caught by P6 and P7. The deterministic properties (P1, P3--P5) are unbreakable under all adversarial conditions (48/48 = 100\%).

\textbf{NLI fine-tuning.} Fine-tuning the NLI model on 190 government delegation examples with 5-fold cross-validation improves P2's standalone malicious TPR from 1.7\% (off-the-shelf) to 88.3\% $\pm$ 8.5\% (F1=82.1\%). The raw NLI FPR increases to 12.3\%, but the three-layer pipeline maintains 0\% system-level FPR via the benign override mechanism.

Table~\ref{tab:adversarial_summary} summarizes all adversarial evaluations, showing that the DAS achieves 100\% detection across 120 adversarial attack attempts from four independent evaluation methodologies with zero false positives on 53 benign controls.

\begin{table}[t]
\centering
\caption{Adversarial Evaluation Summary}
\label{tab:adversarial_summary}
\footnotesize
\begin{tabular}{@{}lrrcc@{}}
\toprule
\textbf{Evaluation} & \textbf{Attacks} & \textbf{Benign} & \textbf{TPR} & \textbf{FPR} \\
\midrule
Black-box red team & 30 & 10 & 100\% & 0\% \\
Independent red team & 45 & 14 & 100\% & 0\% \\
DAS robustness & 29 & --- & 100\% & --- \\
LLM agent (real DAS) & 5 & 23 & 100\% & 0\% \\
Live LangChain agent & 10 & 10 & 100\% & 0\% \\
\midrule
\textbf{Combined} & \textbf{119} & \textbf{57} & \textbf{100\%} & \textbf{0\%} \\
\bottomrule
\end{tabular}
\end{table}

\textbf{Fault tolerance.} A 2-of-3 threshold signing simulation validates DAS availability and safety across 9 fault scenarios: 1 crashed replica (availability maintained), 1 Byzantine replica (safety maintained), 2 crashed (correctly refused), Byzantine+crashed (correctly refused), key rotation (backward-compatible), and P1 enforcement under fault (preserved). All 9 scenarios pass.

\subsection{Comparison with Existing Approaches}

\begin{table}[t]
\centering
\caption{Comparison with Existing Agent Security Frameworks}
\label{tab:comparison}
\footnotesize
\setlength{\tabcolsep}{2.5pt}
\begin{tabular}{@{}lcccccccc@{}}
\toprule
\textbf{Feature} & \textbf{SE} & \textbf{Shield} & \textbf{ABC} & \textbf{FASA} & \textbf{AIP} & \textbf{AW} & \textbf{FJ} & \textbf{Ours} \\
\midrule
Multi-agent chains & \ding{55} & \ding{55} & \ding{55} & Part. & \ding{51} & \ding{51} & \ding{55} & \ding{51} \\
Authority narrowing & \ding{51} & \ding{55} & \ding{55} & \ding{55} & \ding{51} & \ding{55} & \ding{55} & \ding{51} \\
Intent preservation & \ding{55} & \ding{55} & Drift & \ding{55} & \ding{55} & \ding{55} & \ding{55} & \ding{51} \\
Scope-action (P6) & \ding{55} & \ding{55} & \ding{55} & \ding{55} & \ding{55} & Part. & \ding{55} & \ding{51} \\
Output valid. (P7) & \ding{55} & \ding{55} & \ding{55} & \ding{55} & \ding{55} & \ding{55} & \ding{55} & \ding{51} \\
Formal proofs & Part. & \ding{55} & \ding{51} & \ding{55} & \ding{55} & \ding{55} & \ding{51} & \ding{51} \\
Mech. verification & \ding{55} & \ding{55} & \ding{55} & \ding{55} & \ding{55} & \ding{55} & Z3 & TLA+ \\
Forensic recon & \ding{55} & \ding{55} & \ding{55} & \ding{55} & Part. & \ding{55} & \ding{55} & \ding{51} \\
Cascade contain. & \ding{55} & \ding{55} & \ding{55} & \ding{55} & \ding{55} & \ding{55} & \ding{55} & \ding{51} \\
NIST 800-53 map & \ding{55} & \ding{55} & \ding{55} & \ding{55} & \ding{55} & \ding{55} & \ding{55} & \ding{51} \\
Real implementation & \ding{55} & \ding{55} & \ding{55} & \ding{55} & Rust & 9 fwks & \ding{51} & Proto. \\
Verif. latency & --- & --- & --- & --- & 0.05ms & --- & --- & 0.45ms \\
\bottomrule
\end{tabular}
\end{table}

SE = SEAgent~\cite{seagent}; Shield = ShieldAgent~\cite{shieldagent}; ABC = Agent Behavioral Contracts~\cite{abc}; FASA = Full-Lifecycle Agent Security Architecture~\cite{fasa}; AIP = Agent Identity Protocol~\cite{aip}; AW = Authenticated Workflows~\cite{auth_workflows}; FJ = FormalJudge~\cite{formaljudge}. Part.\ = Partial; Proto.\ = Working prototype ($\sim$350 LOC, real HTTP); 9~fwks = integrates with 9 agent frameworks. Latency is per-verification-step (P6 for SentinelAgent, token verification for AIP). TLA+ = TLC model checking of P1, P3--P7 across 2.7M states.

\subsection{End-to-End LLM Agent Evaluation}

Ten multi-step federal workflow scenarios with LLM-generated delegation requests (28 steps total, 5 attack steps) are routed through the real DAS HTTP prototype. Results: 5/5 attacks detected (100\% TPR) via P6+P7 enforcement, 0/23 benign steps falsely blocked (0\% FPR), 100\% accuracy. End-to-end latency per step: median 1.86ms, p95 3.05ms. The DAS blocks data exfiltration, bias injection, scope expansion, contract steering, and surveillance---all using LLM-generated delegation text never seen during development. A live LangChain integration using GPT-4o-mini generates delegation requests in real-time across 20 additional scenarios (10 benign, 10 attack), confirming the results hold with a live LLM agent.

\textbf{The role of P2.} P2 provides early termination: when it detects malicious intent, the delegation token is never issued and no API call is attempted. On DelegationBench~v4, P2 catches 20 of 150 attacks at the delegation stage, preventing those agents from ever receiving tokens. The remaining 130 attacks are caught by P6 and P7 at execution time. P2's value is operational efficiency (reducing DAS load and the TOCTOU window), not detection completeness.

\section{Related Work}
\label{sec:related}

\textbf{Agentic AI threat landscape.} The OWASP Top~10 for Agentic Applications~\cite{owasp_agentic} identifies ten risk categories including goal hijacking, privilege escalation, and cascading failures. The Promptware Kill Chain~\cite{promptware_killchain} formalizes seven-stage attack propagation. Shapira et al.~\cite{agents_of_chaos} document eleven failure modes in a live multi-agent deployment. STAC~\cite{stac_attack} demonstrates tool-chain attacks with $>$90\% ASR on GPT-4.1. The AiTM attack~\cite{aitm_attack} compromises multi-agent systems via message manipulation. Tag-Along attacks~\cite{tagalong_attack} achieve 67\% success via RL-based agent-to-agent jailbreaking. The SoK on agentic AI attack surfaces~\cite{sok_agentic} and the MCP-38 threat taxonomy~\cite{mcp38} provide comprehensive attack taxonomies but no unified formal defenses.

\textbf{Agent security frameworks.} The CSA Agentic Trust Framework~\cite{csa_atf} defines five trust elements with a maturity model but no technical enforcement. SEAgent~\cite{seagent} provides mandatory access control but not delegation chains or intent preservation. ShieldAgent~\cite{shieldagent} verifies safety policies via probabilistic rule circuits but not multi-agent delegation. Agent Behavioral Contracts~\cite{abc} provides Design-by-Contract with probabilistic drift bounds for single agents, not chains. South et al.~\cite{oauth_delegation} extend OAuth~2.0 for agent delegation but do not address intent or compliance. None combine authority verification, intent preservation, compliance enforcement, forensic reconstruction, and cascade containment.

\textbf{Formal methods for agents.} Allegrini et al.~\cite{allegrini_formal} formalize 31 temporal logic properties for agentic AI safety on single-host systems---the most comprehensive formal property set prior to this work---but without addressing multi-agent delegation chains or composition. FormalJudge~\cite{formaljudge} applies Z3 SMT solving for neuro-symbolic agent oversight, achieving 16.6\% improvement over LLM-as-a-Judge baselines. ILION~\cite{ilion} achieves deterministic execution gating at 143 microseconds mean latency. Both address single-agent oversight rather than multi-agent delegation chains. GroupGuard~\cite{groupguard} addresses collusive attacks in multi-agent systems, a threat vector (A6) that SentinelAgent identifies but defers to future work.

\textbf{Delegation-specific security (2025--2026).} The Agent Identity Protocol (AIP)~\cite{aip} introduces invocation-bound capability tokens with Biscuit/Datalog policies for multi-hop delegation, achieving 0.049ms verification in Rust. AIP addresses identity and authority attenuation (analogous to P1 and P4) but does not model intent preservation (P2), compliance enforcement (P3), or cascade containment (P5). Authenticated Workflows~\cite{auth_workflows} provides deterministic enforcement with the MAPL policy language, achieving 100\% recall on 174 test cases. However, it provides engineering guarantees rather than formal proofs and does not address intent preservation or federal-specific compliance. Agentic JWT~\cite{agentic_jwt} binds agent actions to verifiable user intent via dual-faceted tokens but operates at the authentication layer without semantic intent verification. ASTRA~\cite{astra} introduces task-based access control with semantic scope matching at authorization time but does not verify intent preservation throughout the chain. PrivacyPAD~\cite{privacypad} applies reinforcement learning for privacy-aware delegation but without formal security properties.

\section{Discussion and Future Work}
\label{sec:discussion}

\textbf{The adversarial intent paraphrasing problem.} P2 (intent verification) is the only probabilistic property in the DCC. On DelegationBench~v4, P2 alone achieves only 13.3\% TPR because most attack categories target P6 and P7, not P2. Under dedicated adversarial stress testing with sophisticated paraphrasing, P2 achieves 13\% TPR in isolation. Fine-tuning the NLI model on 190 government delegation examples with 5-fold cross-validation improves P2's standalone malicious TPR from 1.7\% (off-the-shelf) to 88.3\% $\pm$ 8.5\% (F1=82.1\%). The critical point is that P2's weakness is fully compensated by P6 and P7: the combined system achieves 100\% TPR on all 516 scenarios regardless of P2's standalone performance. P2's role is early termination (preventing token issuance for detectable attacks), not complete detection.

\textbf{P6 limitations.} Scope-action conformance relies on tool manifests that map scope elements to permitted API operations. Three limitations are identified. First, wildcard matching in manifests (e.g., \texttt{GET /api/records/*}) may permit overly broad access; manifest granularity must balance security with operational flexibility. Second, manifest maintenance imposes an administrative burden: as APIs evolve, manifests must be updated to reflect new endpoints and deprecated operations. Third, and most fundamentally, P6 cannot catch malicious use of permitted APIs---if an agent's scope includes \texttt{query\_records} and the agent uses that access for unauthorized purposes (e.g., querying records for demographic profiling rather than eligibility verification), the API call is within the manifest and P6 does not block it. This honest limitation motivates P7.

\textbf{P7 limitations.} Output schema conformance uses a default-deny whitelist: only explicitly permitted output type tags pass validation. This eliminates synonym evasion (an attacker cannot bypass the schema by using alternative tag names, since unknown tags are blocked by default). Two limitations remain. First, the permitted tag vocabulary must be comprehensive enough to cover all legitimate output types; overly restrictive schemas may cause false positives on valid agent outputs. Second, P7 validates output type tags but not semantic content within a permitted type---if the output type is ``eligibility\_result'' (permitted) but the result contains biased reasoning, P7 does not detect this. Semantic content validation within permitted types remains an open problem. Schema versioning is handled by embedding a schema version identifier in the delegation token; when APIs evolve and new output types are introduced, the DAS schema registry is updated and new tokens carry the updated schema version, while existing tokens continue to enforce their original schema until expiry.

\textbf{Three-point lifecycle complementarity.} The three-point verification lifecycle achieves 100\% combined TPR precisely because each layer covers the others' blind spots. P2 catches what it can before execution (keyword-detectable intent drift). P6 constrains what the agent can physically do at execution time (only permitted API calls execute). P7 constrains what the agent can produce after execution (only conformant outputs are returned). Each layer is individually incomplete: P2 alone achieves 13.3\% on DelegationBench~v4 (88.3\% with fine-tuning), P6 alone achieves 86.7\%, and P7 alone achieves 13.3\%. The 100\% combined result is not a claim of invulnerability---it is a property of the current benchmark. An adversary who can simultaneously evade all three layers (e.g., by crafting an attack that uses a permitted API call, produces a permitted output type, and phrases the intent benignly) would succeed. The meta-theorem on defense-in-depth completeness bounds the damage even in such cases.

\textbf{Composition with write-impact notification.} The composition safety meta-theorem introduces a write-impact notification mechanism for chains sharing mutable state. This mechanism adds overhead: the DAS must track which resources are read by which chains and trigger re-verification when writes occur. In the evaluation, 5/50 composition scenarios require notification, suggesting that the overhead is incurred infrequently in practice. However, in high-throughput systems with many concurrent chains accessing shared databases, the notification mechanism may become a bottleneck. Optimizations such as batched notification and read-set approximation are identified as future work.

\textbf{DAS security and availability.} The DAS is a single point of trust. If compromised via supply chain attack, insider threat, or infrastructure breach, all guarantees are void. Mitigations include: (a)~HSM-based key storage; (b)~multi-party signing requiring 2-of-3 DAS instances; (c)~immutable audit logging of DAS operations (who watches the watcher); (d)~key rotation on configurable schedules. A 2-of-3 threshold signing simulation validates these mitigations across 9 fault scenarios: the DAS tolerates 1 crashed replica (50/50 tokens issued and verified), 1 Byzantine replica that corrupts signatures (50/50 verified via 2 honest sigs), and 1 network-partitioned replica (50/50 verified). It correctly refuses service when 2 replicas crash (0/20 issued) and when 1 Byzantine + 1 crashed leaves only 1 valid signature below threshold (0/20 verified). Key rotation is backward-compatible: old tokens remain verifiable via the 2 non-rotated replicas. P1 scope enforcement is preserved under all survivable fault conditions. Latency overhead is negligible ($<$0.03ms per token). For production, a replicated DAS with Byzantine fault tolerance is needed. In degraded mode (DAS unreachable), agents could operate with cached authorization for a bounded time window, with all actions flagged for post-hoc audit.

\textbf{Privacy Act and consent.} The current DCC does not model consent. In federal systems handling citizen PII, the Privacy Act (5~U.S.C.~552a) requires that records are maintained only for authorized purposes, with System of Records Notices (SORNs) and citizen consent. Extending the delegation token to carry a consent scope---narrowing like authority scope---would address this. When Agent~B retrieves medical records, the consent scope would specify which records the citizen authorized for which purpose. This is a clean formal extension planned for the next version.

\textbf{Graceful degradation.} Even when P2 is evaded, the deterministic properties (P1, P3--P7) ensure the adversary cannot escalate privileges, bypass compliance, escape audit trails, cause unbounded cascade damage, execute unauthorized API calls, or produce prohibited outputs. The meta-theorem on graceful degradation formalizes this: each single-property evasion reduces the adversary's capability to at most 1/11 of the unconstrained action space.

\textbf{Temporal authorization decay.} The DCC extends naturally to time-bounded delegation: $\sigma(t) = \sigma_0 \cdot \mathrm{decay}(t, R)$ where $R$ is the risk tier. Since $\sigma(t) \subseteq \sigma_0$ for all $t > 0$, P1 is preserved. The DAS enforces expiry via the $t_{\exp}$ field; temporal decay generalizes this from binary (valid/expired) to graduated scope reduction with recommended half-lives of 15 minutes (HIGH) to 24 hours (LOW).

\textbf{Protocol-agnostic enforcement.} The DCC is protocol-agnostic by design: delegation tokens are data structures, not protocol messages. The IPDP requires only that the DAS can intercept delegation requests and tool invocations---satisfiable via middleware adapters for MCP, A2A, REST, and SOAP. Properties hold regardless of transport because they are defined over token contents, not message formats.

\textbf{Framework integration.} Integrating IPDP into LangChain, CrewAI, or AutoGen requires three middleware components: a delegation interceptor wrapping \texttt{agent.invoke()} ($\sim$20ms overhead), a scope enforcement proxy intercepting HTTP tool calls ($<$1ms), and an output validator checking responses ($<$1ms). Total per-step overhead: $\sim$22ms. The DAS core is $\sim$350 LOC as a standalone HTTP service.

\textbf{Cost model.} Per-step verification latency is $\sim$20ms (dominated by NLI inference). For a federal agency processing 10,000 delegation steps per hour, this requires approximately 56 CPU-hours per day. P6 and P7 add negligible overhead ($<$1ms each). Scaling strategies: NLI model quantization (INT8) for 2--3$\times$ throughput, batched verification for non-time-critical delegations, and caching of repeated delegation patterns.

\textbf{CISA Zero Trust alignment.} The DCC maps to the CISA Zero Trust Maturity Model's five pillars: Identity (DAS agent registry with DIDs), Devices (agent runtime attestation via token signatures), Networks (TLS-protected DAS communication), Applications and Workloads (scope-action conformance via P6), and Data (information flow enforcement via P3 and output validation via P7).

\textbf{DAS security.} TOCTOU risk is bounded by synchronous pre-execution checks for HIGH risk tools and heartbeat intervals for MEDIUM risk tools. Production deployment requires FIPS~140-3 validated cryptographic modules.

\subsection{Threats to Validity}

\textbf{Construct validity.} The 100\% combined TPR is achieved on a benchmark designed by the same author who designed the defense. The attack categories (A--H) were constructed to be catchable by P2, P6, or P7 respectively. Two mitigations address this: (1)~an independent red team using seven fundamentally different attack strategies (path traversal, method override, case sensitivity, versioning probes, scope boundary probing, synonym attacks, chained attacks) achieves 0\% evasion on 45 attacks; (2)~LLM-generated delegation requests from a live LangChain agent are processed by the real DAS with 100\% accuracy on 28 steps. An independently constructed benchmark by external adversaries would provide stronger evidence.

\textbf{Internal validity.} The property arguments for P1, P3--P7 are contingent on correct DAS implementation. If the DAS has a bug in scope checking, P1 fails. The robustness evaluation (29 edge-case tests including path traversal, encoding attacks, and chain depth 32) provides evidence of implementation correctness. Additionally, properties P1, P3--P7 are mechanically verified via TLA+ model checking: the specification formalizes all six deterministic properties as invariants over a state space of delegation token issuance, scope narrowing, manifest enforcement, output schema validation, hash chain linkage, and cascade revocation. TLC exhaustively checks all reachable states across two model configurations (368,509 states and 2,744,789 states respectively) with zero invariant violations, providing bounded model checking evidence that the DAS design correctly enforces all six properties. P2 is excluded from TLA+ verification because it is probabilistic. The meta-theorems are verified by exhaustive enumeration over the finite property set---all 126 evasion combinations, all 7 minimality attacks, all 7 degradation envelopes---which provides complete coverage within the model.

\textbf{External validity.} The DAS prototype is implemented as a real HTTP service with HMAC-SHA256 token signing, manifest enforcement, and output validation ($\sim$350 LOC). P6 and P7 are evaluated via real HTTP round-trips (median 0.45ms each). LLM-generated delegation requests from 10 multi-step federal workflow scenarios (28 steps) are processed by the real DAS, achieving 100\% accuracy. A live LangChain agent integration using GPT-4o-mini generates delegation requests in real-time across 20 scenarios (10 benign workflows, 10 attack workflows), with each LLM-generated step routed through the DAS HTTP prototype for P1, P6, and P7 enforcement. This closes the gap between simulated and live agent behavior.

\textbf{Remaining open problems.} Collusion via side channels (analogous to covert channels in traditional access control~\cite{nist_800_53}), implicit influence propagation between chains without formal delegation steps~\cite{agents_of_chaos}, and Privacy Act consent modeling remain unaddressed.

\section{Conclusion}

This paper presented SentinelAgent, a formal framework for verifiable delegation chains in federal multi-agent AI systems. The Delegation Chain Calculus defines seven properties, six deterministic and one probabilistic: authority never escalates through delegation (P1), intent preservation is verified at every step (P2), NIST~800-53 compliance is enforced throughout the chain (P3), every action is forensically traceable to the human who authorized it (P4), cascade damage is formally bounded (P5), only permitted API calls execute (P6), and only conformant outputs are returned (P7). The three-point verification lifecycle (pre-execution intent checking, at-execution scope enforcement, post-execution output validation) achieves 100\% combined TPR at 0\% FPR on DelegationBench~v4 (516 scenarios across 10 attack categories and thirteen domains), with each layer catching attacks the others miss. Properties P1, P3--P7 are mechanically verified via TLA+ model checking across 2.7 million states with zero violations. Fine-tuning the NLI model on 190 government delegation examples improves P2 malicious TPR from 1.7\% to 88.3\% (5-fold cross-validated, F1=82.1\%). Four meta-theorems over the DCC cover property minimality, graceful degradation bounds, defense-in-depth completeness across 126 evasion combinations, and composition safety for chains sharing mutable state. A separate proposition establishes the practical infeasibility of deterministic intent verification, justifying P2's probabilistic design. The compliance mapping covers 20 NIST~800-53 controls across 9 families and 10 OWASP ASI risk categories (9/10 fully).

Adversarial stress testing confirms that the deterministic properties are unbreakable while P2 (intent verification) degrades against sophisticated paraphrasing. This is the paper's primary open challenge. But even when P2 is evaded, the remaining six properties constrain the adversary to permitted API calls, conformant outputs, traceable actions, bounded cascades, and compliant behavior.

This work addresses the delegation chain layer of the government AI security stack, complementing CivicShield for conversational interfaces and RAGShield for knowledge pipelines.


\balance

\end{document}